\documentclass{llncs}

\usepackage{epsfig}
\usepackage{gensymb}
\usepackage[english]{babel}
\usepackage{graphicx}
\usepackage{amssymb}
\usepackage{multirow}
\usepackage{xcolor}
\usepackage{xspace}

\usepackage{todonotes}

\newcommand{\remove}[1]{}

\renewcommand{\int}{int}
 \definecolor{blue}{rgb}{0,0,255}
 \definecolor{red}{rgb}{255,0,0}
 \definecolor{green}{rgb}{0,255,0}

\newcommand{\Oh}[1]{\ensuremath{\protect\mathcal{O}(#1)}}
\newcommand{\etal}{\emph{et al.}}

\renewenvironment{proof}
{{\bf Proof:}}{\hspace*{\fill}$\Box$\par\vspace{2mm}}

\begin{document}
\title{Stack and Queue Layouts via \\Layered Separators\thanks{The research of Vida Dujmovi\'c was partially supported by %the Natural Sciences and Engineering Research Council of Canada 
NSERC,  and Ontario Ministry of Research and Innovation. The research of Fabrizio Frati was partially supported by MIUR Project ``AMANDA'' under PRIN 2012C4E3KT.}
}
\author{Vida Dujmovi\'c \inst{1}, Fabrizio Frati\inst{2}}
\institute{1. School of Computer Science and Electrical Engineering, Univ. of Ottawa, Canada\\
\email{vida.dujmovic@uottawa.ca}\\
2. Dipartimento di Ingegneria, Roma Tre Univ., Italy\\
\email{frati@dia.uniroma3.it}\\
}
\maketitle

\begin{abstract}
It is known that every proper minor-closed class of graphs has bounded stack-number (a.k.a.\ book thickness and page number).  While this includes notable graph families such as planar graphs and graphs of bounded genus, many other graph families are not closed under taking minors. For fixed $g$ and $k$, we show that every $n$-vertex graph that can be embedded on a surface of genus $g$ with at most $k$ crossings per edge has stack-number \Oh{\log n}; this includes $k$-planar graphs. The previously best known bound for the stack-number of these families was \Oh{\sqrt{n}}, except in the case of $1$-planar graphs. Analogous results are proved for map graphs that can be embedded on a surface of fixed genus. None of these families is closed under taking minors. The main ingredient in the proof of these results is a construction proving that $n$-vertex graphs that admit constant layered separators have \Oh{\log n} stack-number.
\end{abstract}

\section{Introduction} \label{se:introduction}

A \emph{stack layout} of a graph $G$ consists of a total order $\sigma$ of $V(G)$ and a partition of $E(G)$ into  sets (called \emph{stacks}) such that no two edges in the same stack \emph{cross}; that is, there are no edges $vw$ and $xy$ in a single stack with $v<_\sigma x<_\sigma w<_\sigma y$.  The minimum number of stacks in a stack layout of $G$  is the \emph{stack-number} of $G$. Stack layouts, first defined by Ollmann~\cite{Ollmann73}, are ubiquitous structures with a variety of applications (see \cite{DujWoo-DMTCS04} for a survey). A stack layout is also called a {\em book embedding} and stack-number is also called {\em book thickness} and {\em page number}. The stack-number is known to be bounded for planar graphs \cite{Yannakakis89}, bounded genus graphs \cite{Malitz94b} and, most generally, all proper minor-closed graph families \cite{Blankenship-PhD03,BO01}. 

The purpose of this note is to bring the study of the stack-number beyond the proper minor-closed graph families. 
\emph{Layered separators} are a key tool for proving our results. They have already led to progress on long-standing open problems related to 3D graph drawings~\cite{Duj15,DBLP:conf/focs/DujmovicMW13} and nonrepetitive graph colourings~\cite{DBLP:journals/combinatorics/DujmovicJFW13}.
A {\em layering} $\{V_0,\dots,V_p\}$ of a graph $G$ is a partition of $V(G)$ into {\em layers} $V_i$ such that, for each $e\in E(G)$, there is an $i$ such that the endpoints of $e$ are both in $V_i$ or one in $V_i$ and one in $V_{i+1}$. A graph $G$ has a {\em layered $\ell$-separator} for a fixed layering $\{V_0,\dots,V_p\}$ if, for every subgraph $G'$ of $G$, there exists  a set $S\subseteq V(G')$ with at most $\ell$ vertices in each layer (i.e., $V_i\cap S\leq \ell$, for $i=0,\dots,p$) such that each connected component of $G'-S$ has at most $|V(G')|/2$ vertices. Our main technical contribution is the following theorem. 

\begin{theorem} \label{th:main-sn}
Every $n$-vertex graph that has a layered $\ell$-separator has stack-number at most  $5\ell \cdot \log_2 n$.
\end{theorem}

%We provide two well-known examples of $1$-planar graphs that contain arbitrarily large complete graph minors. First, one can start with the $n\times n\times 2$ grid graph, which is a $\Theta(n^2)$-vertex $1$-planar graph, and contract the $i$-th row in the front $n\times n$ grid together with the $i$-th column in the back $n\times n$ grid, for every $1\leq i\leq n$; this gives a $K_n$ minor. Second, one can start with a straight-line drawing of $K_n$ in which no three edges pass through the same point and insert a dummy vertex between every two consecutive crossings on each edge; this gives a drawing of a $\Theta(n^4)$-vertex $1$-planar graph that contains $K_n$ as a subdivision.

We discuss the implications of Theorem~\ref{th:main-sn} for two well-known non-minor-closed classes of graphs.  A graph is \emph{$(g,k)$-planar} if it can be drawn on a surface of Euler genus at most $g$ with at most $k$ crossings per edge. Then $(0,0)$-planar graphs are {\em planar graphs}, whose stack-number is at most $4$~\cite{Yannakakis89}. Further, $(0,k)$-planar graphs are {\em $k$-planar graphs}~\cite{PachToth-Comb97}; Bekos~\etal~\cite{DBLP:conf/esa/BekosB0R15} have recently proved that $1$-planar graphs have bounded stack-number (see Alam~\etal~\cite{oneplanar-2} for an improved constant). The family of $(g,k)$-planar graphs is not closed under taking minors\footnote{The $n\times n\times 2$ grid graph is a well-known example of $1$-planar graph with an arbitrarily large complete graph minor. Indeed, contracting the $i$-th row in the front $n\times n$ grid with the $i$-th column in the back $n\times n$ grid, for $1\leq i\leq n$, gives a $K_n$ minor.} even for $g=0$, $k=1$; thus the result of Blankenship and Oporowski \cite{Blankenship-PhD03,BO01}, stating that proper minor-closed graph families have bounded stack-number, does not apply to $(g,k)$-planar graphs. Dujmovi\'c~\etal~\cite{DBLP:conf/gd/DujmovicEW15} showed that $(g,k)$-planar graphs have layered $(4g+6)(k+1)$-separators\footnote{More precisely, Dujmovi\'c~\etal~\cite{DBLP:conf/gd/DujmovicEW15} proved that $(g,k)$-planar graphs have {\em layered treewidth} at most $(4g+6)(k+1)$ and $(g,d)$-map graphs have layered treewidth at most $(2g+3)(2d+1)$. Just as the graphs of treewidth $t$ have (classical) separators of size $t-1$, so do the graphs of layered treewidth $\ell$ have layered $\ell$-separators \cite{DBLP:conf/focs/DujmovicMW13,DMW-2014}.}. This and our Theorem~\ref{th:main-sn} imply the following corollary. For all $g\geq 0$ and $k\geq 2$, the previously best known bound was \Oh{\sqrt{n}}, following from the \Oh{\sqrt{m}} bound for $m$-edge graphs \cite{Malitz94a}.

\begin{corollary} \label{cor:gk}
For any fixed $g$ and $k$, every $n$-vertex $(g,k)$-planar graph has stack-number \Oh{\log n}.
\end{corollary}

A $(g,d)$-map graph $G$ is defined as follows. Embed a graph $H$ on a surface of Euler genus $g$ and label some of its faces as ``nations'' so that any vertex of $H$ is incident to at most $d$ nations; then the vertices of $G$ are the faces of $H$ labeled as nations and the edges of $G$ connect nations that share a vertex of $H$. The $(0,d)$-map graphs are the well-known \emph{$d$-map graphs}~\cite{FLS-SODA12,Chen-JGT07,DFHT05,CGP02,Chen01}. The $(g,3)$-map graphs are the graphs of Euler genus at most $g$~\cite{CGP02}, thus they are closed under taking minors. However, for every $g\geq 0$ and $d\geq 4$, the $(g,d)$-map graphs are not closed under taking minors \cite{DBLP:conf/gd/DujmovicEW15}, thus  the result of Blankenship and Oporowski \cite{Blankenship-PhD03,BO01} does not apply to them. The $(g,d)$-map graphs have layered $(2g+3)(2d+1)$-separators \cite{DBLP:conf/gd/DujmovicEW15}. This and our Theorem~\ref{th:main-sn} imply the following corollary. For all $g\geq 0$ and $d\geq 4$, the best previously known bound was \Oh{\sqrt{n}}~\cite{Malitz94a}.

\begin{corollary} \label{cor:gd}
For any fixed $g$ and $d$, every $n$-vertex $(g,d)$-map graph has stack-number \Oh{\log n}.
\end{corollary}

A ``dual'' concept to that of stack layouts are queue layouts. A \emph{queue layout} of a graph $G$ consists of a total order $\sigma$ of $V(G)$ and a partition of $E(G)$ into sets (called \emph{queues}), such that no two edges in the same queue \emph{nest}; that is, there are no edges $vw$ and $xy$ in a single queue with $v<_\sigma x<_\sigma y<_\sigma w$. If $v<_\sigma x<_\sigma y<_\sigma w$ we say that $xy$ \emph{nests inside} $vw$. The minimum number of queues in a queue layout of $G$ is called the \emph{queue-number} of $G$. Queue layouts, like stack layouts, have been extensively studied. In particular, it is a long standing open problem to determine if planar graphs have bounded queue-number. Logarithmic upper bounds have been obtained via layered separators~\cite{Duj15,layer-pw}. In particular, a result similar to Theorem~\ref{th:main-sn} is known for the queue-number: Every $n$-vertex graph that has layered $\ell$-separators has queue-number \Oh{\ell\log n}~\cite{Duj15}; this bound was refined to $3\ell \cdot \log_3(2n + 1)-1$ by Bannister~\etal~\cite{layer-pw}.
 These results were established via a connection with the {\em track-number} of a graph~\cite{DMW-SJC05}. Together with the fact that planar graphs have layered $2$-separators~\cite{DBLP:journals/combinatorics/DujmovicJFW13,LT-SJAM79}, these results imply an \Oh{\log n} bound for the queue-number of planar graphs, improving on a earlier result by Di Battista~\etal~\cite{SCICOMP-BattistaFP13}. The polylog bound on the queue-number of planar graphs extends to all proper minor-closed families of graphs ~\cite{DBLP:conf/focs/DujmovicMW13,DMW-2014}. Our approach to prove Theorem~\ref{th:main-sn} also gives a new proof of the following result (without using track layouts). We include it for completeness.

\begin{theorem} \label{th:main-qn}
Every $n$-vertex graph that has a layered $\ell$-separator has queue-number at most  $3\ell \cdot \log_2 n$. %<1.76 \ell\cdot \log_{3/2} n$.
\end{theorem}

\section{Proofs of Theorem \ref{th:main-sn} and Theorem \ref{th:main-qn}}

Let $G$ be a graph and $L=\{V_0,\dots,V_p\}$ be a layering of $G$ such that $G$ admits a layered $\ell$-separator for layering $L$. Each edge of $G$ is either an \emph{intra-layer} edge, that is, an edge between two vertices in a set $V_i$, or an \emph{inter-layer} edge, that is, an edge between a vertex in a set $V_i$ and a vertex in a set $V_{i+1}$. 

A total order on a set of vertices $R\subseteq V(G)$ is a \emph{vertex ordering} of $R$. The stack layout construction computes a vertex ordering $\sigma^s$ of $V(G)$ satisfying the {\em layer-by-layer} invariant, which is defined as follows: For $0\leq i < p$, the vertices in $V_i$ precede the vertices in $V_{i+1}$ in $\sigma^s$. Analogously, the queue layout construction computes a vertex ordering  $\sigma^q$ of $V(G)$ satisfying the layer-by-layer invariant. 

Let $S$ be a layered $\ell$-separator for $G$ with respect to $L$. Let $G_1,\dots,G_k$ be the graphs induced by the vertices in the connected components of $G-S$ (the vertices of $S$ do not belong to any graph $G_j$). These graphs are labeled  $G_1,\dots,G_k$ arbitrarily. Recall that, by the definition of a layered $\ell$-separator for $G$, we have $|V(G_j)|\leq n/2$, for each $1\leq j\leq k$. Let $S_i=S\cap V_i$ and let $\rho_i$ be an arbitrary vertex ordering of $S_i$, for $i=0,\dots,p$.

Both the stack and the queue layout constructions recursively construct vertex orderings of $V(G_j)$ satisfying the layer-by-layer invariant, for $j=1,\dots,k$. Let $\sigma^s_j$ be the vertex ordering of $V(G_j)$ computed by the stack layout construction; we also denote by $\sigma^s_{j,i}$ the restriction of $\sigma^s_j$ to the vertices in layer $V_i$. Note that $\sigma^s_j=\sigma^s_{j,1},\sigma^s_{j,2},\dots,\sigma^s_{j,p}$ by the layer-by-layer invariant. Vertex orderings $\sigma^q_j$ and $\sigma^q_{j,i}$ are defined analogously for the queue layout construction.

We now show how to combine the recursively constructed vertex orderings to obtain a vertex ordering of $V(G)$. The way this combination is performed differs for the stack layout construction and the queue layout construction.

{\bf Stack layout construction.} Vertex ordering $\sigma^s$ is defined as (refer to Fig.~\ref{fig:stack-layout}) 
\begin{eqnarray*}
& & \rho_0,\sigma^s_{1,0},\sigma^s_{2,0},\dots,\sigma^s_{k-1,0},\sigma^s_{k,0},\rho_1,\sigma^s_{k,1},\sigma^s_{k-1,1},\dots,\sigma^s_{2,1},\sigma^s_{1,1},\\
& &
\rho_2,\sigma^s_{1,2},\sigma^s_{2,2},\dots,\sigma^s_{k-1,2},\sigma^s_{k,2},\rho_3,\sigma^s_{k,3},\sigma^s_{k-1,3},\dots,\sigma^s_{2,3},\sigma^s_{1,3},\dots. 
\end{eqnarray*}

\begin{figure}[tb]
    \centering
	\includegraphics[width=0.99\textwidth]{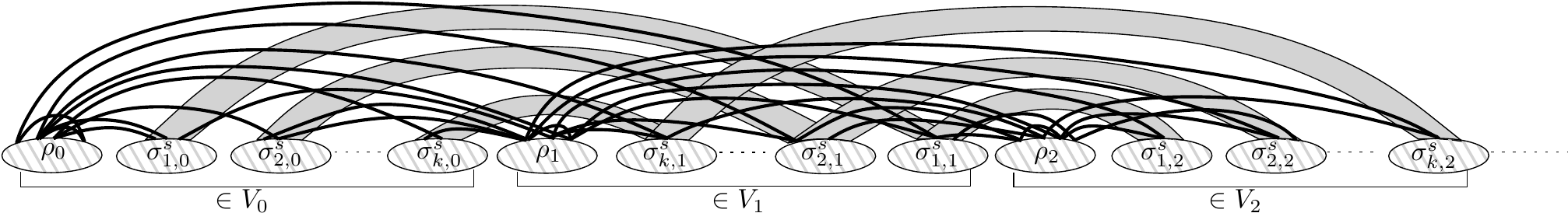}
    \caption{Illustration for the stack layout construction. Edges incident to vertices in $S$ are black and thick. Edges in graphs $G_1,\dots,G_k$ are represented by gray regions.}
\label{fig:stack-layout}
\end{figure}

The vertex ordering $\sigma^s$ satisfies the layer-by-layer invariant, given that vertex ordering $\sigma^s_j$ does, for $j=1,\dots,k$. Then Theorem~\ref{th:main-sn} is implied by the following.

\begin{lemma} \label{le:stacks}
$G$ has a stack layout with $5\ell \cdot \log_2 n$ stacks with vertex ordering~$\sigma^s$.
\end{lemma} 

\begin{proof}
We use distinct sets of stacks for the intra- and the inter-layer edges. 

{\em Stacks for the intra-layer edges.} We assign each intra-layer edge $uv$ with $u\in S$ or $v\in S$ to one of $\ell$ stacks $P_1,\dots,P_\ell$ as follows. Since $uv$ is an intra-layer edge, $\{u,v\}\subseteq V_i$, for some $0\leq i\leq p$. Assume w.l.o.g. that $u<_{\sigma^s}v$. Then $u\in S$ and let it be $x$-th vertex in $\rho_i$ (recall that $\rho_i$ contains at most $\ell$ vertices). Assign $uv$ to $P_x$.  The only intra-layer edges that are not yet assigned to stacks belong to graphs $G_1,\dots,G_k$. The assignment of these edges to stacks is the one computed recursively; however, we use the same set of stacks to assign the edges of all graphs $G_1,\dots,G_k$. 

We now prove that no two intra-layer edges in the same stack cross. Let $e$ and $e'$ be two intra-layer edges of $G$ and let both the endpoints of $e$ be in $V_i$ and both the endpoints of $e'$ be in $V_{i'}$. Assume w.l.o.g. that $i\leq i'$. If $i<i'$, then, since $\sigma^s$ satisfies the layer-by-layer invariant, the endpoints of $e$ precede those of $e'$ in $\sigma^s$, hence $e$ and $e'$ do not cross. 
Suppose now that $i=i'$. 
If $e$ and $e'$ are in some stack $P_x$ for $x\in \{1, \dots, \ell\}$, then they are both incident to the $x$-th vertex in $\rho_i$, thus they do not cross. If $e$ and $e'$ are in some stack different from $P_1,\dots,P_\ell$, then $e\in E(G_j)$ and $e'\in E(G_{j'})$, for some $j,j'\in \{1,\dots,k\}$. If $j=j'$, then $e$ and $e'$ do not cross by induction. Otherwise, both the endpoints of $e$ precede both the endpoints of $e'$ or vice versa, since the vertices in $\sigma^s_{\min\{j,j'\},i}$ precede  those in $\sigma^s_{\max\{j,j'\},i}$ in $\sigma^s$ or vice versa, depending on whether $i$ is even or odd; hence $e$ and $e'$ do not cross.

We now bound the number of stacks we use for the intra-layer edges of $G$; we claim that this number is at most $\ell \cdot \log_2 n$. The proof is by induction on $n$; the base case $n=1$ is trivial. For any subgraph $H$ of $G$, let $p_1(H)$ be the number of stacks we use for the intra-layer edges of $H$, and let $p_1(n')=\max_H\{p_1(H)\}$ over all subgraphs $H$ of $G$ with $n'$ vertices. As proved above, $p_1(G)\leq \ell + \max\{p_1(G_1),\dots,p_1(G_k)\}$. Since each graph $G_j$ has at most $n/2$ vertices, we get that $p_1(G)\leq \ell + p_1(n/2)$. By induction $p_1(G) \leq \ell + \ell \cdot \log_2 (n/2) = \ell \cdot \log_2 n$.

{\em Stacks for the inter-layer edges.} We use distinct sets of stacks for the {\em even inter-layer edges} -- connecting vertices on layers $V_i$ and $V_{i+1}$ with $i$ even -- and for the {\em odd inter-layer edges} -- connecting vertices on layers $V_i$ and $V_{i+1}$ with $i$ odd. We only describe how to assign the even inter-layer edges to $2\ell \cdot \log_2 n$ stacks so that no two edges in the same stack cross; the assignment for the odd inter-layer edges is analogous.

We assign each even inter-layer edge $uv$ with $u\in S$ or $v\in S$ to one of $2\ell$ stacks $P'_1,\dots,P'_{2\ell}$ as follows. Since $uv$ is an inter-layer edge, $u$ and $v$ respectively belong to layers $V_i$ and $V_{i+1}$, for some $0\leq i\leq p-1$. If $u\in S$, then $u$ is the $x$-th vertex in $\rho_i$, for some $1\leq x \leq \ell$; assign edge $uv$ to $P'_x$. If $u\notin S$, then $v\in S$ is the $y$-th vertex in $\rho_{i+1}$, for some $1\leq y \leq \ell$; assign edge $uv$ to $P'_{\ell + y}$. The only even inter-layer edges that are not yet assigned to stacks belong to graphs $G_1,\dots,G_k$. The assignment of these edges to stacks is the one computed recursively; however, we use the same set of stacks to assign the edges of all graphs $G_1,\dots,G_k$. 

We prove that no two even inter-layer edges in the same stack cross. Let $e$ and $e'$ be two even inter-layer edges of $G$. Let $V_i$ and $V_{i+1}$ be the layers containing the endpoints of $e$.  Let $V_{i'}$ and $V_{i'+1}$ be the layers containing the endpoints of $e'$. Assume w.l.o.g. that $i\leq i'$.  If $i<i'$, then $i+1<i'$, given that both $i$ and $i'$ are even. Then, since $\sigma^s$ satisfies the layer-by-layer invariant, both the endpoints of $e$ precede both the endpoints of $e'$, thus $e$ and $e'$ do not cross. Suppose now that $i=i'$. If $e$ and $e'$ are in some stack $P'_h$ for $h\in\{1, \dots, 2\ell\}$, then $e$ and $e'$ are both incident either to the $h$-th vertex of $\rho_i$ or to the $(h-\ell)$-th vertex of $\rho_{i+1}$, hence they do not cross. If $e$ and $e'$ are in some stack different from $P'_1,\dots,P'_{2\ell}$, then $e\in E(G_j)$ and $e'\in E(G_{j'})$, for $j,j'\in \{1,\dots,k\}$. If $j=j'$, then $e$ and $e'$ do not cross by induction. Otherwise, $j\not= j'$ and then $e$ nests inside $e'$ or vice versa, since the vertices in  $\sigma^s_{\min\{j,j'\},i}$ precede those in  $\sigma^s_{\max\{j,j'\},i}$ and the vertices in $\sigma^s_{\max\{j,j'\},i+1}$ precede those in  $\sigma^s_{\min\{j,j'\},i+1}$ in $\sigma^s$; hence $e$ and $e'$ do not cross.

We now bound the number of stacks we use for the even inter-layer edges of $G$; we claim that this number is at most $2\ell \cdot \log_2 n$. The proof is by induction on $n$; the base case $n=1$ is trivial. For any subgraph $H$ of $G$, let $p_2(H)$ be the number of stacks we use for the even inter-layer edges of $H$, and let $p_2(n')=\max_H\{p_2(H)\}$ over all subgraphs $H$ of $G$ with $n'$ vertices. As proved above, $p_2(G)\leq 2\ell + \max\{p_2(G_1),\dots,p_2(G_k)\}$. Since each graph $G_j$ has at most $n/2$ vertices, we get that $p_2(G)\leq 2\ell + p_2(n/2)$. By induction $p_2(G) \leq 2\ell + 2\ell \cdot \log_2 (n/2) = 2\ell \cdot \log_2 n$.

The described stack layout uses $\ell \cdot \log_2 n$ stacks for the intra-layer edges, $2\ell \cdot \log_2 n$ stacks for the even inter-layer edges, and $2\ell \cdot \log_2 n$ stacks for the odd inter-layer edges, thus $5\ell \cdot \log_2 n$ stacks in total. This concludes the proof.
\end{proof}

{\bf Queue layout construction.} Vertex ordering $\sigma^q$ is defined as (refer to Fig.~\ref{fig:queue-layout}) $\rho_0,\sigma^q_{1,0},\sigma^q_{2,0},\dots,\sigma^q_{k,0},\rho_1,\sigma^q_{1,1},\sigma^q_{2,1},\dots,\sigma^q_{k,1},\dots,\rho_p,\sigma^q_{1,p},\sigma^q_{2,p},\dots,\sigma^q_{k,p}$.

\begin{figure}[tb]
    \centering
	\includegraphics[width=0.99\textwidth]{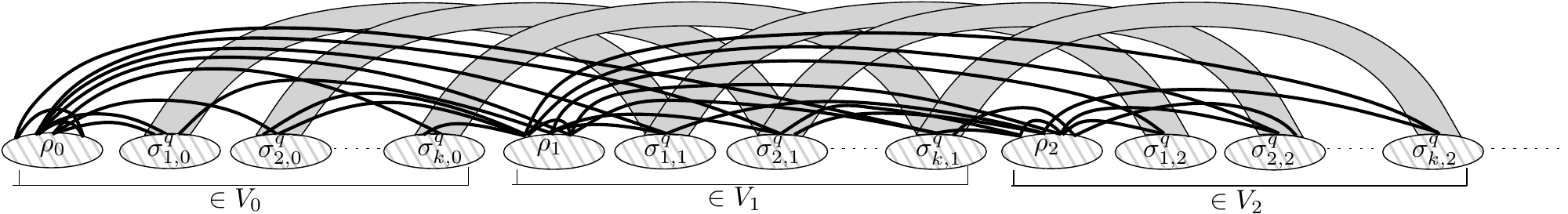}
    \caption{Illustration for the queue layout construction.}
\label{fig:queue-layout}
\end{figure}

The vertex ordering $\sigma^q$ satisfies the layer-by-layer invariant, given that vertex ordering $\sigma^q_j$ does, for $j=1,\dots,k$. Then Theorem~\ref{th:main-qn} is implied by the following.

\begin{lemma} \label{le:queues}
$G$ has a queue layout with $3\ell \cdot \log_2 n$ queues with vertex ordering~$\sigma^q$.
\end{lemma} 

\begin{proof}
We use distinct sets of queues for the intra- and the inter-layer edges. 

{\em Queues for the intra-layer edges.} We assign each intra-layer edge $uv$ with $u\in S$ or $v\in S$ to one of $\ell$ queues $Q_1,\dots,Q_\ell$ as follows. Since $uv$ is an intra-layer edge, $\{u,v\}\subseteq V_i$, for some $0\leq i\leq p$. Assume w.l.o.g. that $u<_{\sigma^q} v$. Then $u\in S$ and let it be the $x$-th vertex of $\rho_i$. Assign $uv$ to $Q_x$. The only intra-layer edges that are not yet assigned to queues belong to graphs $G_1,\dots,G_k$. The assignment of these edges to queues is the one computed recursively; however, we use the same set of queues to assign the edges of all graphs $G_1,\dots,G_k$. 

The proof that no two intra-layer edges in the same queue nest is the same as the proof no two intra-layer edges in the same stack cross in Lemma~\ref{le:stacks} (with the word ``nest'' replacing ``cross'' and with $\sigma^q$ replacing $\sigma^s$). The proof that the number of queues we use for the intra-layer edges is at most $\ell \cdot \log_2 n$ is also the same as the proof that the number of stacks we use for the intra-layer edges is at most $\ell \cdot \log_2 n$ in Lemma~\ref{le:stacks}.

{\em Queues for the inter-layer edges.} We assign each inter-layer edge $uv$ with $u\in S$ or $v\in S$ to one of $2\ell$ queues $Q'_1,\dots,Q'_{2\ell}$ as follows. Since $uv$ is an inter-layer edge, $u$ and $v$ respectively belong to layers $V_i$ and $V_{i+1}$, for some $0\leq i\leq p-1$. If $u\in S$, then $u$ is the $x$-th vertex in $\rho_i$, for some $1\leq x \leq \ell$; assign edge $uv$ to $Q'_x$. If $u\notin S$, then $v\in S$ is the $y$-th vertex in $\rho_{i+1}$, for some $1\leq y \leq \ell$; assign edge $uv$ to $Q'_{\ell + y}$. The only inter-layer edges that are not yet assigned to queues belong to graphs $G_1,\dots,G_k$. The assignment of these edges to queues is the one computed recursively; however, we use the same set of queues to assign the edges of all graphs $G_1,\dots,G_k$.

We prove that no two inter-layer edges $e$ and $e'$ in the same queue nest. Let $V_i$ and $V_{i+1}$ be the layers containing the endpoints of $e$.  Let $V_{i'}$ and $V_{i'+1}$ be the layers containing the endpoints of $e'$. Assume w.l.o.g. that $i\leq i'$.  If $i<i'$, then both endpoints of $e$ precede the endpoint of $e'$ in $V_{i'+1}$ (hence $e'$ is not nested inside $e$) and both endpoints of $e'$ follow the endpoint of $e$ in $V_{i}$ (hence $e$ is not nested inside $e'$), since $\sigma^q$ satisfies the layer-by-layer invariant; thus $e$ and $e'$ do not nest.
Suppose now that $i=i'$. 
If $e$ and $e'$ are in some queue $Q'_h$ for $h\in\{1, \dots, 2\ell\}$, then $e$ and $e'$ are both incident either to the $h$-th vertex of $\rho_i$ or to the $(h-\ell)$-th vertex of $\rho_{i+1}$, hence they do not nest. 
If $e$ and $e'$ are in some queue different from $Q'_1,\dots,Q'_{2\ell}$, then $e\in E(G_j)$ and $e'\in E(G_{j'})$, for $j,j'\in \{1,\dots,k\}$. If $j=j'$, then $e$ and $e'$ do not nest by induction. Otherwise, $j\not= j'$ and then the endpoints of $e$ alternate with those of $e'$ in $\sigma^q$, since the vertices in  $\sigma^q_{\min\{j,j'\},i}$ precede those in  $\sigma^q_{\max\{j,j'\},i}$ and the vertices in  $\sigma^q_{\min\{j,j'\},i+1}$ precede those in  $\sigma^q_{\max\{j,j'\},i+1}$ in $\sigma^q$; hence $e$ and $e'$ do not nest.

We now bound the number of queues we use for the inter-layer edges of $G$; we claim that this number is at most $2\ell \cdot \log_2 n$. The proof is by induction on $n$; the base case $n=1$ is trivial. For any subgraph $H$ of $G$, let $q(H)$ be the number of queues we use for the inter-layer edges of $H$, and let $q(n')=\max_H\{q(H)\}$ over all subgraphs $H$ of $G$ with $n'$ vertices. As proved above, $q(G)\leq 2\ell + \max\{q(G_1),\dots,q(G_k)\}$. Since each graph $G_j$ has at most $n/2$ vertices, we get that $q(G)\leq 2\ell + q(n/2)$. By induction $q(G) \leq 2\ell + 2\ell \cdot \log_2 (n/2) = 2\ell \cdot \log_2 n$.

Thus, the described queue layout uses $\ell \cdot \log_2 n$ queues for the intra-layer edges and $2\ell \cdot \log_2 n$ queues for the inter-layer edges, thus $3\ell \cdot \log_2 n$ queues in total. This concludes the proof.
\end{proof}

\noindent{\bf Acknowledgments:} The authors wish to thank David R. Wood for stimulating discussions and comments on the preliminary version of this article.

\bibliographystyle{splncs03}
\bibliography{bibliography}

\end{document}